\documentclass{article}
        
\let\it=\itshape          
\let\bf=\bfseries

\usepackage{amsmath}
\usepackage{amssymb}
\usepackage{color}
\usepackage{graphicx}
\usepackage{psfrag}
\usepackage{epsfig}
\usepackage{comment}

\newcommand{\Fig}[1]{Fig.~\ref{#1}}

\newtheorem{lemma}{Lemma}
\newtheorem{definition}{Definition}
\newtheorem{corollary}{Corollary}
\newtheorem{theorem}{Theorem}
\newenvironment{proof}{
   {\bf Proof}}{\hbox{\ }\hfill$|||$
}
\newcommand{\gs}{\sigma}
\newcommand{\lin}{\pmb{\ell}}
\newcommand{\drv}{\ell}
\newcommand{\pt}{\mathbf{p}}
\newcommand{\tn}{\mathbf{t}}
\newcommand{\al}[2]{a^{#1}_{#2}}
\newcommand{\alq}[2]{a_{#2,#1}}
\def\R{\mathbb{R}}
\def\ii{k}
\def\IL{{\it left}}
\def\IM{{\it middle}}
\def\IR{{\it right}}

\def\nor{\mathbf{n}}
\def\B{\mathbf{b}}

\def\pat#1{\B^{#1}}

\def\CAGDfig{.}

\newcommand{\val}{n}
\newcommand{\bv}{\mathbf{v}}
\newcommand{\numer}{\beta}
\newcommand{\crv}{\mathbf{c}}

\newcommand{\denom}{\gamma}
\newcommand{\dg}[1]{\text{deg}(#1)}
\newcommand{\genspl}{generalized spline}
\newcommand{\nurbs}{spline}
\newcommand{\geo}{geometric design}
\newcommand{\kl}{knot line}
\newcommand{\ik}{edge knot}
\newcommand{\ibv}{edge vertex}
\newcommand{\vloc}{vertex-localized}

\title{On the Complexity of Smooth Spline Surfaces from Quad Meshes}
\author{ J\"org Peters and Jianhua Fan
        \\
        University of Florida
        }
\begin{document}
\maketitle
\begin{abstract}
This paper derives strong relations that \emph{boundary curves} of
a smooth complex of patches
have to obey when the patches are computed by local averaging.
These relations restrict the choice of reparameterizations
for geometric continuity.

In particular,
when one bicubic tensor-product B-spline patch is associated
with each facet of a quadrilateral mesh 
with $n$-valent vertices and we do not want segments of
the boundary curves forced to be linear,
then the relations dictate the minimal number and 
multiplicity of knots: For general data, the tensor-product spline patches
must have at least two internal double knots per edge to be able to 
model a $G^1$-conneced complex of $C^1$ splines.
This lower bound on the complexity of any construction
is proven to be sharp by suitably interpreting an existing 
surface construction.
That is, we have a tight bound on the complexity of smoothing quad meshes
with bicubic tensor-product B-spline patches.
%
\end{abstract}
\section{Introduction}
\label{sec:Introduction}
Even though every newly proposed smooth surface construction
seeks to be optimal in some aspect, the overall theory of 
smooth surface constructions offers few sharp
lower bounds, i.e.\ proofs that no polynomial construction 
of lower degree is possible and that a construction of this
least degree exists so that upper bound and lower bound match.
One well-appreciated bound is the degree-6 bound for
$C^2$ subdivision surfaces derived by Reif and Prautzsch 
\cite{journals/cagd/Prautzsch97} and shown to be sharp,
for example by \cite{Reif:1998:turbs,Prautzsch-Reif:1998}.
Such sharp bounds allow us to \\
--- understand the fundamental difficulty of the task, and to \\
--- guide future research by showing where research is futile \\
--- and what assumptions must be side-stepped to derive substantially new
results.

We are motivated by a standard task of \geo:
to determine $G^1$-connected tensor-product B-spline patches  
approximating a quadrilateral mesh whose vertices can have any fixed valence.
While this challenge can be met by recursive subdivision \cite{Catmull-1978-CC},
representing the surface with a finite small number of patches
defined by the quad and its neighbors is often preferable, for example
to parallelize the construction (see e.g.\ \cite{Loop:2008:ACC,Myles:2008:GCQ}).
This raises the question: (Q)
\emph{what is the simplest structure} (in distribution and number of knots)
\emph{of degree bi-3 \nurbs\ patches that allow a quad mesh
to be converted by localized operations
into a smooth surface with one \nurbs\ patch per quad?}
Surprisingly, this basic question at the heart of a classical task
of \geo\ has not been settled to date.

To frame the question, Section \ref{sec:g1} takes a more general view.
We do not constrain the domain to be a collection of
quadrilaterals or the functions to be polynomial splines.
Also, the relations in
Lemmas \ref{lem:n4xform}, \ref{lem:c1spline} and \ref{lem:c2spline} 
do not depend on locality of the construction
but apply to any collection of sufficiently smooth patches
coming together with a logically symmetric $G^1$ join:
$
   \partial_2\pat{\ii}(u,0)
   +\partial_1\pat{\ii-1}(0,u)
   =
   \alpha^{\ii}(u)
   \partial_1 \pat{\ii}(u,0)
$ (see Definition \ref{def:g1}, page \pageref{def:g1}). 
Adding locality of operations as a requirement
in Section \ref{subsec:local} then rules out everywhere (piecewise) linear
$\alpha^{\ii}$, still in the very general setting.

In Section \ref{sec:lower}, we specialize the setting
to polynomial tensor-product splines of degree bi-3.
For these, we obtain a lower bound on the number and
multiplicity of knots. We prove that 
at least two internal double knots are required per edge 
to admit a local construction.
This lower bound is tight, because the recently-published
construction for smooth surfaces \cite{Fan:2008:SBS}
can be re-interpreted as a spline construction with exactly
two internal double knots.
Together, the lower and upper bound conclusively settle the question Q.

\subsection{Bi-3 constructions in the literature}
Creating $C^1$ surfaces with a finite number of patches of degree  bi-3,
i.e.\ generalizing standard tensor-product B-splines to 
smooth surfaces from arbitrary manifold quad meshes,
is a classic challenge of CAGD 
(see e.g.\ \cite{bezier77a,Wijk:1986:BPA,Peters:1991:SIM}).
The assumption that a simple construction with a finite number of patches 
is not possible motivates the classic Catmull-Clark subdivision
(\Fig{fig:pccm}, \IL).
PCCM \cite{PCCM} is a finite construction that approximates
Catmull-Clark limit surfaces with smoothly connected bi-3 patches.
PCCM requires up to two steps of Catmull-Clark subdivision to separate
non-4-valent vertices. This proves that a $4\times4$ arrangement of 
polynomial patches per quad suffices in principle,
corresponding to two double interior knots and one single knot
(\Fig{fig:pccm}, \IM), 
However, PCCM can have poor shape for certain higher-order saddles 
(\Fig{fig:defect}, \cite{url:saddle,Peters:2000:ModPCC,Loop:2008:ACC}).
\def\wid{\linewidth}
\begin{figure}[h]
    \centering
    \psfrag{CC}{Catmull-Clark}
    \psfrag{PCCM}{PCCM}
    \psfrag{Fan}{\cite{Fan:2008:SBS}}
    \epsfig{file=\CAGDfig/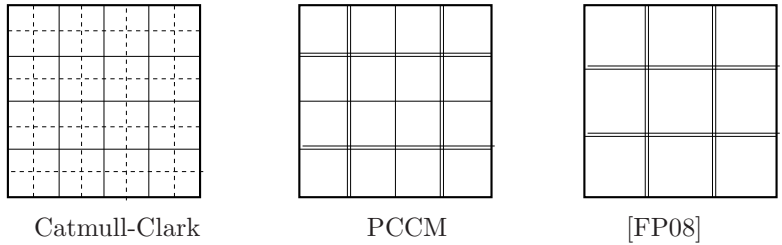,width=0.85\wid}
    \caption{
    {\bf Knot distribution.} A quadrilateral piece generated by  Catmull-Clark 
    subdivision has (infinitely many) single knots,
    a piece of PCCM requires two double and at least one more single knot,
    and the construction \cite{Fan:2008:SBS} has two double interior knots
    (which this paper shows to be the minimal number of knots).
    }
    \label{fig:pccm}
\end{figure}
More recently, a number of papers appeared that are 
also predicated on the assumption that a simple construction
with a finite number of patches is not possible.
Shi et al. \cite{journals/cagd/ShiWWL04,shi2006rcg} propose 
a subdivision-like refinement approach  with 
bi-3 tensor-product patches to obtain $C^0$ surfaces 
where ever more single knots are inserted.
They correctly surmise that, in general, no finite
$C^1$ construction with $C^2$ tensor-product splines of degree
bi-3 is possible (see Theorem \ref{thm:twodblknt} of our paper).
At the other extreme, using a single patch per quad,
Loop and Schaefer \cite{Loop:2008:ACC} propose a bi-3 $C^0$ surface
construction with separate tangent patches to convey an impression of
smoothness as in \cite{Vlachos:2000:CPNT}, while
Myles \emph{et al.} \cite{Myles:2008:GCQ} perturb a
bi-3 base patch near non-4-valent vertices to
obtain a $C^1$ surface of degree bi-5 for CAD applications.
Hahmann et al. \cite{conf/gmp/HahmannBC08} propose
a $2\times2$ macro-patch per quad; and
Fan and Peters \cite{Fan:2008:SBS} present an algorithm
that constructs smoothly connected B\'ezier patches of 
degree bi-3 whose internal transitions 
allow re-interpretion 
as one tensor-product spline patch per quad
with two internal double knots (\Fig{fig:pccm}, \IR, 
Corollary \ref{cor:opt}).
We will see that this is indeed
the minimal number and multiplicity of knots for 
the standard Catmull-Clark layout of patches.
The structurally different polar layout
allows collapsed bi-3 spline patches with single internal knots
to complete a $C^1$ surface \cite{Myles:2007:ECC}.

\section{Unbiased $G^1$ constraints}
\label{sec:g1}
We consider $\val$ parameterically $C^1$ patches 
\begin{equation}
   \pat{\ii} :  \Box \subsetneq \R^2 \to \R^3,\qquad
   \ii=1,\ldots,\val
\end{equation}
meeting at a central point $\pat{\ii}(0,0)=\pt$
such that $\pat{\ii}(u,0)= \pat{\ii-1}(0,u)$
(see \Fig{fig:crvnet}).
We do not (yet) assume that $\Box$ is the unit square
but just that the origin is a corner of the domain $\Box$
and that two edges emanate from it in independent directions.
We also assume that the patches are not 
singular at the origin in the sense that 
$\partial_2\pat{\ii}(0,0) \times \partial_1 \pat{\ii}(0,0) \ne 0$
where $\partial_\ell$ 
denotes differentiation with respect to the $\ell$th argument.

\def\wid{\linewidth}
\begin{figure}[h]
    \hskip1cm
    \psfrag{u}{$_1$}
    \psfrag{v}{$_2$}
    \psfrag{ci}{$\pat{\ii}(u,0)= \pat{\ii-1}(0,u)$}
    \psfrag{xi}{$\pat{\ii-1}$}
    \psfrag{xp}{$\pat{\ii}$}
    \psfrag{p}[r]{$\pt$}
    \epsfig{file=\CAGDfig/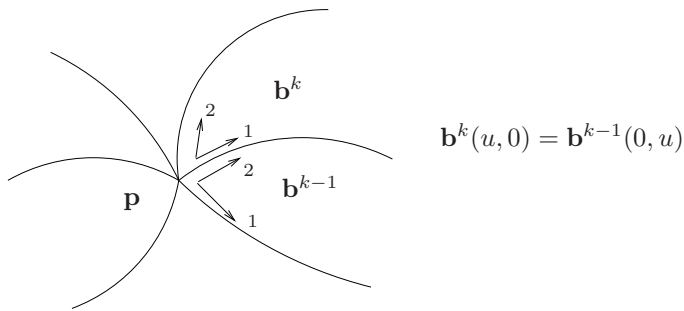,width=0.8\wid}
    \caption{
    {\bf Indexing and parameterization} of adjacent patches
    at a vertex of valence $\val$
    (if $\ii=1$ then $\pat{\ii-1}=\pat{\val}$), illustrating 
    the $G^1$ constraints \eqref{eq:g1}
    .
    }
    \label{fig:crvnet}
\end{figure}
To make the $\val$ patches form a $C^1$ surface, we want to enforce 
logically symmetric (unbiased) $G^1$ constraints.
(We will discuss the general case in Section \ref{sec:discuss}.)
\begin{definition}[Unbiased $G^1$ constraints]
With $\alpha^{\ii}: \R \to \R$ a sufficiently smooth,
univariate scalar-valued function,
the \emph{unbiased} $G^1$ constraints between consecutive patches are
\begin{equation}
   \partial_2\pat{\ii}(u,0)
   +\partial_1\pat{\ii-1}(0,u)
   =
   \alpha^{\ii}(u)
   \partial_1 \pat{\ii}(u,0)
   .
   \label{eq:g1}
\end{equation}
If $\alpha^{\ii}\equiv 0$, 
the constraints enforce \emph{parametric $C^1$ continuity}. \\
We abbreviate 
\begin{equation}
   \al{\ii}{\drv} \in \R,
   \qquad
   \text{ the $\drv$th derivative of }  \alpha^\ii
   \text{ evaluated at }  0 
\end{equation}
and 
\begin{equation}
   \tn^\ii := \partial_1 \pat{\ii}(0,0)  \in \R^3
\end{equation}
so that relation \eqref{eq:g1} becomes at $(0,0)$
\begin{equation*}
\raisebox{0.9cm}{
$
    \tn^{\ii+1}+\tn^{\ii-1} = \al{\ii}{0}\tn^{\ii}
    .
   \quad
   \eqref{eq:g1}_{u=0}$}
    \hskip1cm
    \psfrag{a0k}{$\al{\ii}{0}$}
    \psfrag{bm}{$\tn^{\ii+1}$}
    \psfrag{bu}[l]{$\tn^{\ii}$}
    \psfrag{bp}{$\tn^{\ii-1}$}
    \psfrag{eqn}[c]{}
    \epsfig{file=\CAGDfig/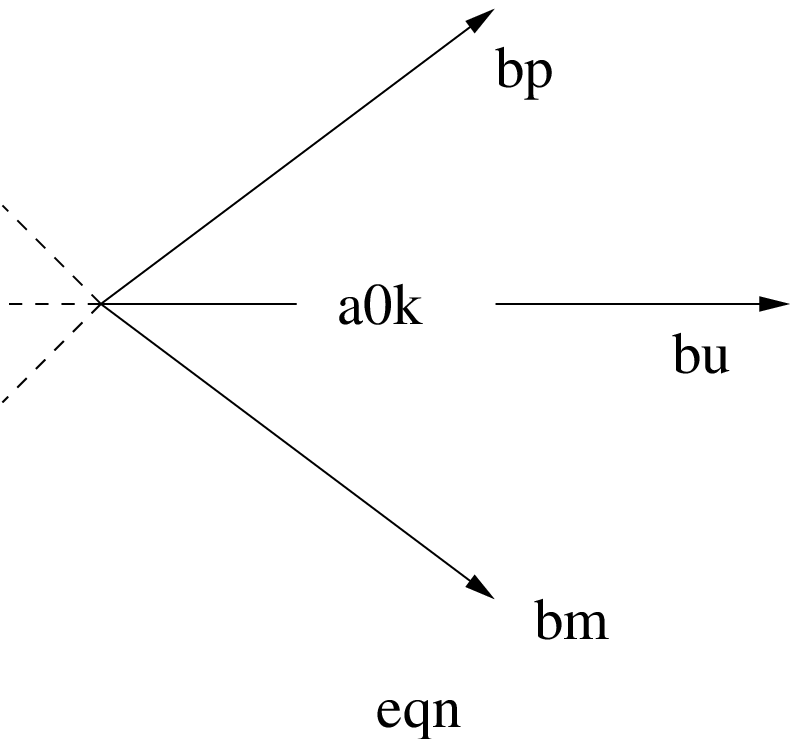,width=0.25\wid}
\end{equation*}
\label{def:g1}
\end{definition}
\vskip-.5cm
That is, superscripts count sectors (modulo $\val$) surrounding
$(0,0)$ while subscripts indicate derivatives. 
Later, starting with \eqref{eq:alphajell}, we will use a second subscript
(and remove the superscript) to denote pieces of $\alpha^\ii$.

\medskip
We now add the assumption that each $\pat{\ii}$ is twice 
continuously differentiable at $(0,0)$
(as are the polynomial pieces of a \nurbs\ patch).
In reference to the main application,
we will call such smooth functions \genspl s.

\begin{definition}[Knot lines and \genspl s]
A $C^s$ \emph{\genspl} patch is a map 
$\pat{\ii}: \Box \subsetneq \R^2 \to \R^3$
that is $s$ times continuously differentiable.
The set of \emph{\kl s} of $\pat{\ii}$ is a
finite collection of lines in $\Box$ such that at most two distinct
lines cross. Boundary edges of $\Box$ are \kl s.
An intersection of a boundary edge minus its end points
with a non-parallel \kl\ 
is called an \emph{\ik}.
At every \ik, 
on either side of its \kl, 
\begin{align}
   \partial^i_1\partial^j_2\pat{\ii} \text{ is well-defined for } i+j\le s+1
   \quad
   \text{ and } 
   \quad
   \partial_2\pat{\ii} \times \partial_1 \pat{\ii} \ne 0.
\end{align}
\label{def:spline}
\end{definition}

The \genspl\ definition is intentionally broader than its
subclass of polynomial tensor-product splines that motivates it.
It includes, for example, trigonometric splines or subdivision constructions.

For $C^1$ \genspl s, we can then differentiate relation \eqref{eq:g1}
along (the respective domain edge of)
the common boundary $\pat{\ii}(u,0)=\pat{\ii-1}(0,u)$:
\begin{align}
   \notag
   &(\partial_1\partial_2\pat{\ii})(u,0)
   +
   (\partial_2\partial_1\pat{\ii-1})(0,u)
   \\
   &=
   \alpha^\ii(u) \partial^2_1 \pat{\ii}(u,0)
   +
   (\alpha^{\ii})'(u) \partial_1 \pat{\ii}(u,0)
   .
   \label{eq:g11gen}
\end{align}
When we evaluate at $u=0$ then 
\begin{equation}
   \text{ at } (0,0),\qquad
   \partial_1\partial_2\pat{\ii}
   +
   \partial_2\partial_1\pat{\ii-1}
   =
   \al{\ii}{0} \partial^2_1 \pat{\ii}
   +
   \al{\ii}{1} \partial_1 \pat{\ii}
   .
   \label{eq:g11gen0}
\end{equation}
If $\val$ is \emph{even} then the alternating sum of the left hand
sides vanishes
\begin{equation}
   \text{ at } (0,0),\qquad
   \sum^\val_{\ii=1} (-1)^\ii
   \bigl(
   \partial_1\partial_2\pat{\ii}
   +
   \partial_2\partial_1\pat{\ii-1}
   \bigr)
   =0 
   \label{eq:g11evenleft}
\end{equation}
and therefore so must the right hand side
\begin{equation}
   \text{ at } (0,0),\qquad
   0 =
   \sum^\val_{\ii=1} (-1)^\ii
   \al{\ii}{0} \partial^2_1 \pat{\ii}
   +
   \sum^\val_{\ii=1} (-1)^\ii
   \al{\ii}{1} \partial_1 \pat{\ii}
   .
   \label{eq:g11evenright}
\end{equation}
In particular, if the patches join smoothly and therefore have a unique
normal $\nor\in \R^3$ at $\pt$ then, with $\cdot$ denoting the scalar product,
\begin{equation}
   \text{ if $\val$ is even, at (0,0) }\qquad
   0 =
   \sum^\val_{\ii=1} (-1)^\ii
   \al{\ii}{0}\ \nor \cdot \partial^2_1 \pat{\ii}
   .
   \label{eq:vertencl}
\end{equation}
This is the \emph{vertex-enclosure constraint}
(see e.g.\ \cite[p.205]{Peters:2002:GC}).

We briefly focus on the important generic case where $\val=4$ patches meet.
\begin{definition} [tangent X]
If $\val=4$, 
$\partial_1 \pat{1}(0,0)=-\partial_1 \pat{3}(0,0)$ 
and $\partial_1 \pat{2}(0,0)=-\partial_1 \pat{4}(0,0)$
then the tangents form an X.
\end{definition}

\begin{lemma} [X tangent]
If the tangents form an X, then \\
$\al{1}{1}=\al{3}{1} \text{ and } \al{2}{1}=\al{4}{1}. $
\label{lem:n4xform}
\end{lemma}
\begin{proof}
If the tangents form an X then $\val=4$ and $\al{\ii}{0}=0$,
$\ii=1,2,3,4$ so that \eqref{eq:g11evenright} simplifies to 
\begin{align}
   \label{eq:2x2vert}
   \text{ at (0,0)},\qquad
   &0 =
   (\al{1}{1}-\al{3}{1}) \partial_1 \pat{1}
   -
   (\al{2}{1}-\al{4}{1}) \partial_1 \pat{2}
   .
\end{align}
Since the patches are regular at corners,
both summands have to vanish, implying the claim.
\end{proof}

We now consider the unbiased $G^1$ transition between 
two $C^1$ \genspl\ patches. 
We focus on an \emph{\ibv}, the image of an 
\emph{\ik} on the common boundary. 
By definition, an \ibv\ is not an end point of the boundary.
That is, we consider a point where four polynomial 
pieces meet such that $\pat{1}$ and $\pat{2}$ belong to 
one \genspl\ patch and $\pat{3}$ and $\pat{4}$ are adjacent pieces
of the edge-adjacent \genspl\ patch (Figure \ref{fig:spl2}).
Since each \genspl\ patch is internally parameterically $C^1$,
by Definition \ref{def:g1}
\begin{equation}
   \alpha^2\equiv 0 \equiv \alpha^4.
   \label{eq:al0}
\end{equation}
\begin{figure}[h]
    \centering
    \psfrag{b1}[l]{$\pat{1}$}
    \psfrag{b2}{$\pat{2}$}
    \psfrag{b3}{$\pat{3}$}
    \psfrag{b4}[l]{$\pat{4}$}
    \psfrag{a1}{$\alpha^1$}
    \psfrag{a2}{$\alpha^2$}
    \psfrag{a3}{$\alpha^3$}
    \psfrag{s1}{$_1$}
    \psfrag{s2}{$_2$}
    \epsfig{file=\CAGDfig/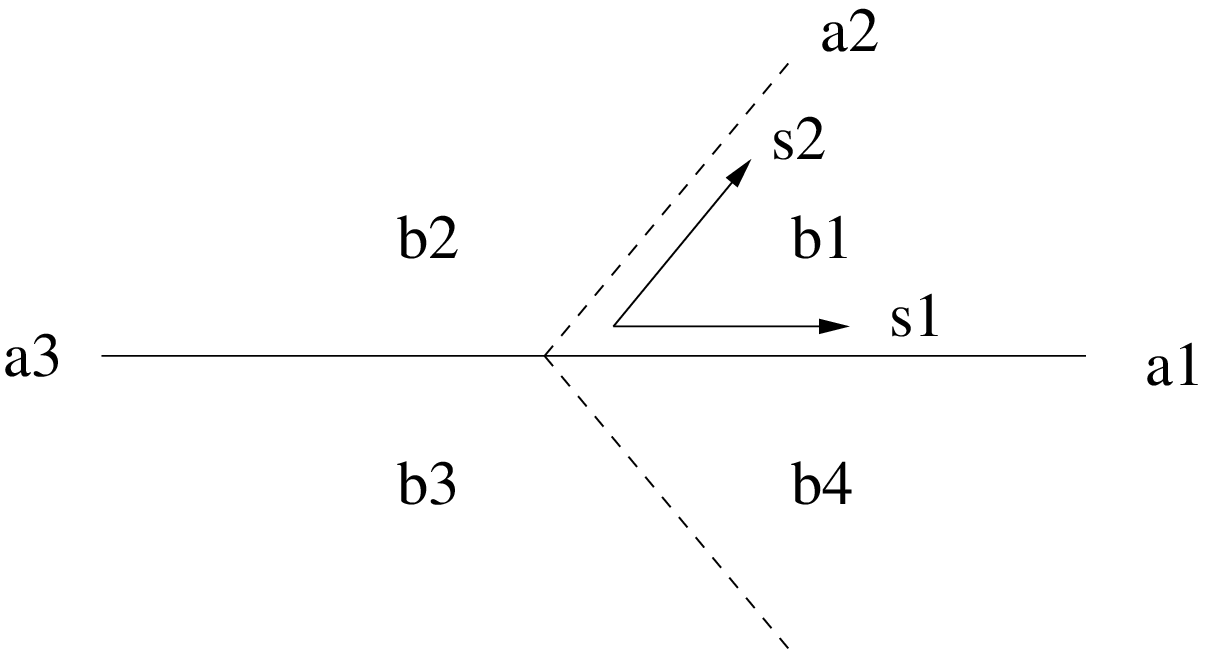,width=0.5\wid}
    \caption{
    Join across an {\bf \ik } on the boundary 
    (solid) between two \genspl s.
    The first \genspl\  has polynomial pieces $\pat{1}$ and $\pat{2}$.
    }
    \label{fig:spl2}
\end{figure}
\begin{lemma} [$C^1$ \genspl, \ik]
Let $(0,0)$ be the parameter
associated with an \ik\ on the boundary common to two
$C^1$ \genspl s that are joined by unbiased $G^1$ constraints. Then
\begin{align}
   \al{1}{0} &= -\al{3}{0}, 
   \label{eq:al10}
   \\
   \text{at } (0,0):\
   0 &=
   \al{1}{0} ( \partial^2_1 \pat{1} - \partial^2_1 \pat{3})
   + (\al{1}{1} - \al{3}{1})\tn^1
   .
   \label{eq:dd}
\end{align}
\label{lem:c1spline}
\end{lemma}
\begin{proof}
Since $\val=4$,
$\al{1}{0}\tn^1 = \tn^2 +\tn^4 =  \al{3}{0}\tn^3$
and the parametric $C^1$ constraints imply $\tn^1 := -\tn^3$ 
so that \eqref{eq:al10} follows.
By \eqref{eq:al0}, \eqref{eq:g11evenright} specializes to
\begin{align}
   \notag
   \text{ at } (0,0),\qquad
   0 &=
   \al{1}{0} \partial^2_1 \pat{1}
   +
   \al{3}{0} \partial^2_1 \pat{3}
   +
   \al{1}{1} \partial_1 \pat{1}
   +
   \al{3}{1} \partial_1 \pat{3}
   \\
   \notag
   &=
   \al{1}{0} ( \partial^2_1 \pat{1} - \partial^2_1 \pat{3})
   + (\al{1}{1} - \al{3}{1})\tn^1
\end{align}
as claimed.
\end{proof}
\\
So, remarkably, when two \genspl\ patches meet along a common boundary,
unbiased $G^1$ constraints \emph{across} this boundary imply 
the constraint \eqref{eq:dd} exclusively in terms of derivatives
\emph{along} the boundary.
\begin{lemma} [$C^2$ \genspl, \ik]
Let $(0,0)$ be the parameter
associated with an \ik\ of the boundary common to two
$C^2$ \genspl s joined by unbiased $G^1$ constraints.
Then, in addition to \eqref{eq:al10}, at $(0,0)$,
\begin{align}
   \al{1}{1} &= \al{3}{1},\ 
   \label{eq:al11}
   \\
   0 &=
   \al{1}{0} ( \partial^3_1 \pat{1} - \partial^3_1 \pat{3})
   + 4 \al{1}{1} \partial^2_1 \pat{1} 
   + (\al{1}{2} - \al{3}{2})\tn^1
   .
   \label{eq:ddd}
\end{align}
\label{lem:c2spline}
\end{lemma}
\begin{proof}
Since the \genspl s are $C^2$,
$\partial^2_1 \pat{1}(0,0) = \partial^2_1 \pat{3} (0,0)$.
Then \eqref{eq:dd} implies \eqref{eq:al11}.

Parametric $C^2$ continuity across the spline-internal boundaries
(see dashed lines in \Fig{fig:spl2}) implies
\begin{equation}
\text{ for } \ii=2,4, \text{ at } (0,0),\qquad
   \partial_2\partial_1\partial_2\pat{\ii}
   +
   \partial_1\partial_2\partial_1\pat{\ii-1}
   =
   0
   .
   \label{eq:c12}
\end{equation}
Differentiating \eqref{eq:g11gen} once more along the
(direction corresponding to the) common boundary of the two \genspl s,
we obtain  for $\ii=1,3, \text{at } (0,0),$
\begin{equation}
   \partial_1\partial_1\partial_2\pat{\ii}
   +
   \partial_2\partial_2\partial_1\pat{\ii-1}
   =
   \al{\ii}{0} \partial^3_1 \pat{\ii}
   +
   2
   \al{\ii}{1} \partial^2_1 \pat{\ii}
   +
   \al{\ii}{2} \partial_1 \pat{\ii}
   .
   \label{eq:g12gen}
\end{equation}
Summing the two instances of \eqref{eq:g12gen} 
and subtracting the two instances of \eqref{eq:c12} 
eliminates the mixed derivatives of the left hand side
and yields at $(0,0)$
\begin{align}
   0
   &=
   \al{1}{0} \partial^3_1 \pat{1}
   +
   2
   \al{1}{1} \partial^2_1 \pat{1}
   +
   \al{1}{2} \partial_1 \pat{1}
   \\
   \notag
   &+
   \al{3}{0} \partial^3_1 \pat{3}
   +
   2
   \al{3}{1} \partial^2_1 \pat{3}
   +
   \al{3}{2} \partial_1 \pat{3}
   .
\end{align}
Parametric $C^2$ continuity then implies \eqref{eq:ddd}.
\end{proof}

\subsection{Linear $\alpha$ and \vloc\ constructions}
\label{subsec:local}
The Taylor expansions up to order two of the patches joining at a point 
are strongly intermeshed by Equation \eqref{eq:g11gen0}.
To avoid solving large, global systems, 
a vertex should not depend on the expansions at its
neighbors.

\begin{definition} [\vloc\ construction]
A construction is $G^1$ \emph{\vloc} if
we can solve at every vertex (with local parameters $(u,v)=(0,0)$)
the unbiased $G^1$ constraints
\eqref{eq:g1}$_{u=0}$ and \eqref{eq:g11gen} 
on the second-order Taylor expansion
$\partial^i_1\partial^j_2\pat{\ii}$, $0\le i,j, i+j\le 2$
independent of the expansions at its neighbors.
\end{definition}

Note that a \vloc\ construction can use 
{\it a priori} known input, for example the local connectivity 
and the valence of the neighbors.
Nevertheless, the unbiased $G^1$ constraints imply
a \emph{local, unbiased choice of the 
tangent directions}, namely such that 
\begin{equation}
   \alpha^\ii(0) := 2\cos\frac{2\pi}{\val}
   .
\label{eq:unbiasedvert}
\end{equation}
(For a proof that logical symmetry implies 
\eqref{eq:unbiasedvert} see e.g.\ \cite[Prop 3]{Peters:1994:CCM}.)

\begin{corollary} [valence symmetry for $\val=4$ and linear $\alpha$]
Let $\val=4$ and let $n^\ii$ denote the valence of the $\ii$th
neighbor vertex, $\ii=1,\ldots,\val$.
Then a local, unbiased choice of the tangent directions and
$\alpha^\ii$ linear are compatible with 
unbiased $G^1$ constraints only when the valences 
of opposite neighbors agree: $n^\ii=n^{\ii+2}$.
\label{cor:valfour}
\end{corollary}
\begin{proof}
The claim follows from Lemma \ref{lem:n4xform}
since by the unbiased choice
$\alpha^\ii(0) := 0$ and $\alpha^\ii(1) := 
2\cos\frac{2\pi}{\val^k}$.
\end{proof}

Corollary \ref{cor:valfour} is a remarkably strong restriction
since vertices of valence $\val=4$ are common.
Choosing linear $\alpha$ can therefore be problematic.
For example, the construction \cite{conf/gmp/HahmannBC08}
can therefore not succeed in general.

\bigskip
In the most challenging case, 
the vertex enclosure constraint \eqref{eq:vertencl} 
applies at each vertex. 
While the vertex-enclosure constraint 
only restricts the normal component of the second derivatives
along the curves at each vertex, independence of the normals 
at endpoints means that in general all three coordinates are constrained.
So for \vloc\ constructions, we should 
assume that the \emph{second-order Taylor expansion}
has to be set independently at each vertex.
It is this scenario with 
\emph{unrestricted choice of geometric instantiation}
of the second-order Taylor expansion that we take into account
when, in the following, we prefix a statement with {\bf `in general'}.

\begin{figure}[h]
    \centering
    \psfrag{a1=0}{$\alq{0}{0}=0$}
    \psfrag{a2=0}{$\alq{q}{0}=0$}
    \psfrag{a=0}{$\alq{j}{0}=0$}
    \psfrag{amj0}{$\alq{-j}{0}=0$}
    \psfrag{am10}{$\alq{-1}{0}=0$}
    \psfrag{a10}{$\alq{1}{0}=0$}
    \epsfig{file=\CAGDfig/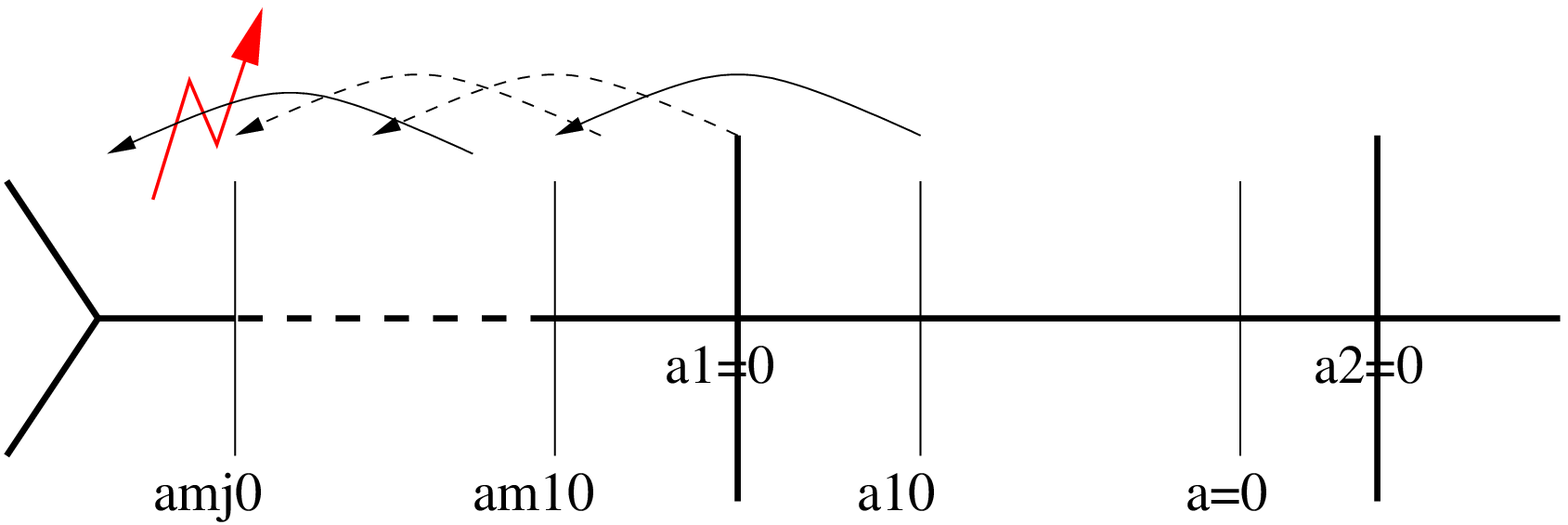,width=\wid}
    \caption{
    {\bf Propagation} of $\alq{j}{0}=0$ in Lemma  \ref{lem:vloc}.
    }
    \label{fig:propag4}
\end{figure}
\medskip
\noindent
\emph{Along a boundary} curve, each scalar function $\alpha^\ii$
can consist of pieces that correspond to the knot segments
of the two \genspl s meeting along the curve.
Since, in this context, we only deal with one $\ii$ at a time,
we drop this $\ii$ superscript and define 
\begin{equation}
   \alq{j}{\ell} \in \R
   \quad
   \text{ to be the }
   \ell\text{th derivative of the }
   j\text{th piece } \alpha_j \text{ at } 0
   .
   \label{eq:alphajell}
\end{equation}
For example, $\alpha^3$ and $\alpha^1$ in \Fig{fig:spl2},
can be relabeled $\alpha_j := \alpha^3$ and $\alpha_{j+1} := \alpha^1$.

\begin{lemma} [everywhere piecewise linear $\alpha$ ruled out]
In general, 
a \vloc\ construction of unbiased 
$G^1$ transitions between $C^1$ \genspl\ patches
with \emph{everywhere} at most linear $\alpha$
is not possible.
\label{lem:vloc}
\end{lemma}
\begin{proof}
Consider a vertex surrounded by vertices of valence $\val=4$.
Then \vloc\ construction implies that $\alq{0}{0}=0$.
Assume for now that \emph{\ik s} exist.
Then local construction implies
also $\alq{-1}{0}=0$ (and $\alq{1}{0}=0$) for the
immediate neighbor \ibv es
since all neighbor vertices are of valence 4 
(by definition at most two \kl s intersect within a \genspl).
Shifting the focus to one such an \ibv, say the one corresponding
to $\alq{-1}{0}=0$, we observe that 
its tangents form an X since the two \genspl s, one at either side,
are internally parametrically $C^1$ (across each dashed line in
\Fig{fig:spl2}).
So Lemma \ref{lem:n4xform} and $\alq{0}{0}=0$ imply $\alq{-2}{0}=0$ 
and again this \ibv's tangents form an X.
In this manner, X configurations and $\alq{-j}{0}=0$ propagate,
also across vertices of valence $\val=4$ whose neighbors are not 
all of valence  $\val=4$
(see the arrows in Figure \ref{fig:propag4} for illustration).
Once the propagation meets an original vertex
with valence $\val\ne 4$
(whether or not we had \ik s to start with), \vloc\ construction
clashes with Lemma \ref{lem:n4xform}.
\end{proof}

Lastly, we characterize a known source of poor shape 
of smooth surface constructions due to restricted boundary curves
\cite{Peters:2000:ModPCC}.
This limited flexibility is undesirable and constructions that cause it
will later be excluded.
\begin{figure}[h]
    \centering
    \epsfig{file=\CAGDfig/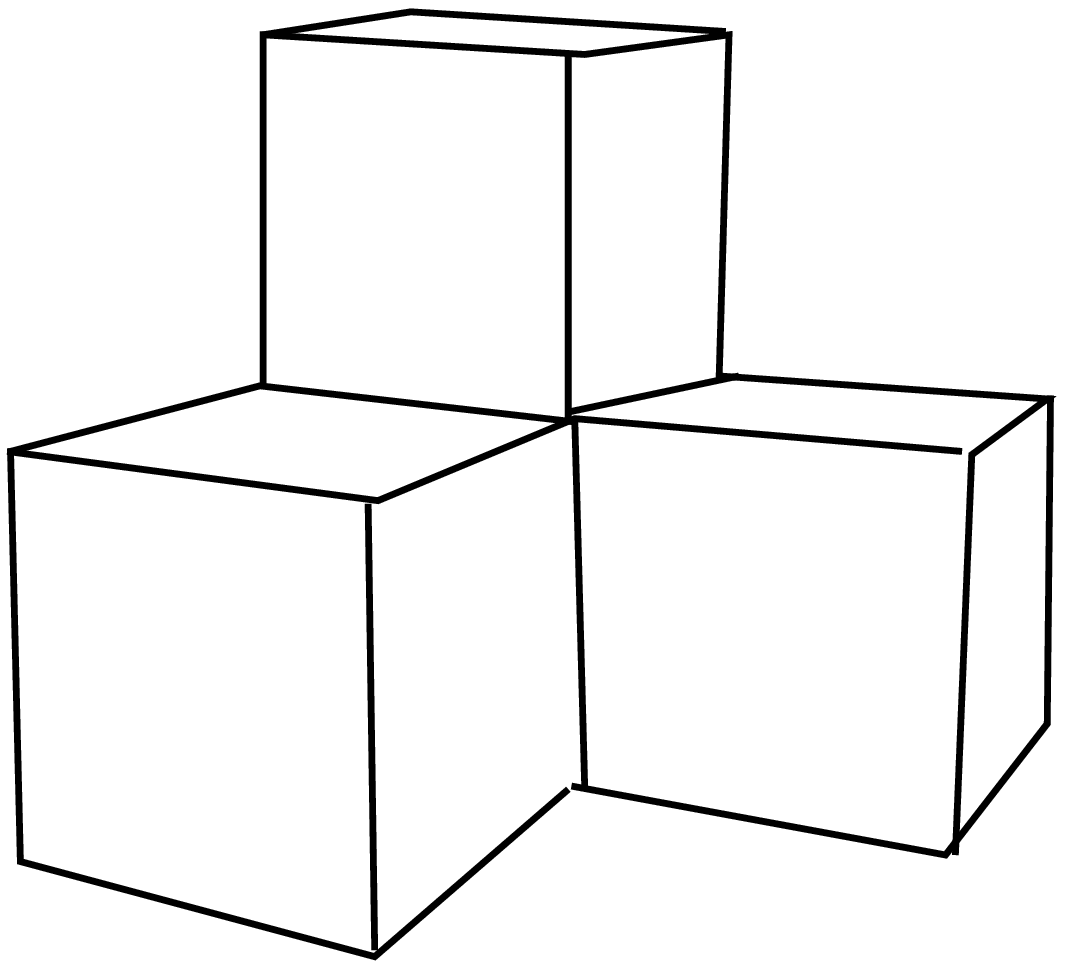,width=0.25\wid}
    \hskip 0.1\wid
    \epsfig{file=\CAGDfig/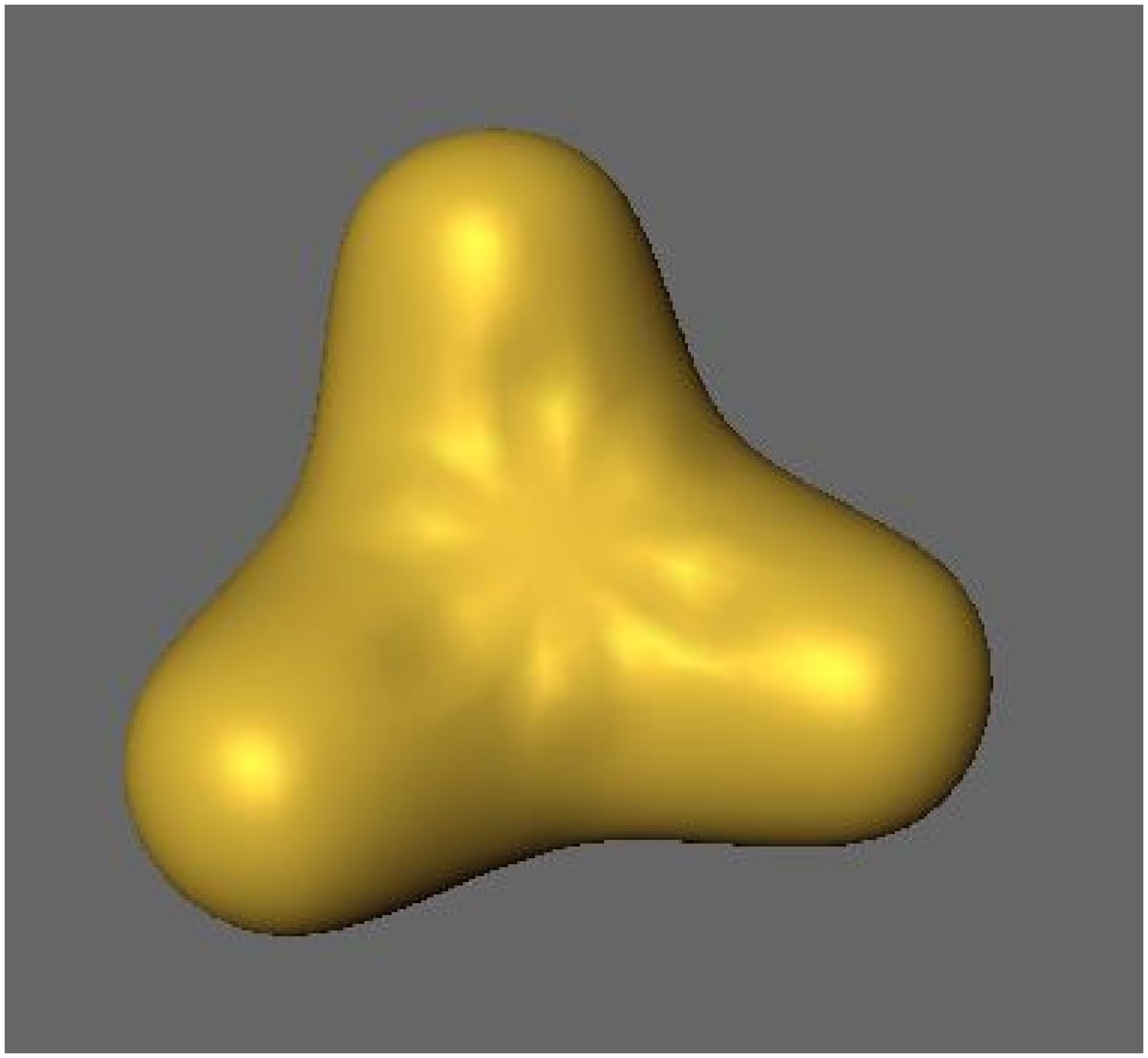,width=0.3\wid}
    \caption{
    {\bf Shape defect} (star shape) due to embedded straight line segments
    at a higher order saddle from \cite{url:saddle}.
    }
    \label{fig:defect}
\end{figure}

\begin{lemma}[Flatness at saddle points]
Let $\crv$ be a curve segment 
emanating from a higher-order
saddle point $\pt := \crv(0)$.
If the derivative $\crv'$ of $\crv$
factors into a linear vector-valued polynomial
and a scalar factor:
\begin{align}
   \crv' &:= \lin\denom,
   \label{eq:lin_den}
   \\ 
   \notag
   \lin&:\R\to\R^3,\ \dg{\lin}\le 1, 
   \quad
   \denom:\R\to\R,\ \dg{\denom}\le 1
\end{align}
then $\crv$ is a planar curve segment.
If the saddle is symmetric
then $\crv$ is a \emph{straight line segment}.
\label{lem:quad}
\end{lemma}
\begin{proof}
Let $\nor$ be the normal at $\pt$
and, without loss of generality, $\denom(u) := 1+\denom_1 u$
for some $\denom_1\in \R$.
Then 
$
   \crv'(0) 
   = \lin(0)
$, 
$
   \crv''(0) 
   = \lin'(0)+\lin(0)\denom_1
$ 
and
$
   \crv'''(0) 
   = 2\lin'(0)\denom_1
$.
At a higher-order saddle point,
the normal curvature is zero, and therefore $\nor\cdot\crv''(0)=0$. 
This implies $\nor\cdot\lin'(0)=0$ and $\nor\cdot\crv'''(0)=0$
establishing planarity.
If the saddle is symmetric then
$\crv'(0)$  and $\crv''(0)$ 
are collinear and so is $\crv'''(0) = 2\lin'(0)\denom_1 $.
\end{proof}

A higher-order saddle, such as the monkey saddle of \Fig{fig:defect}, 
should have non-zero Gauss curvature apart from
the central saddle point. 
Therefore, we will in the following \emph{disqualify} constructions that force 
\emph{straight segments} on the boundary for non-flat geometry.

To summarize, we showed that \vloc\ unbiased $G^1$ constructions
with \genspl s are subject to strong restrictions on the 
reparametrization $\alpha$ (Lemma 
\ref{lem:n4xform}, \ref{lem:c1spline} and \ref{lem:c2spline})
or the allowable valence of the 
vertices (Corollary \ref{cor:valfour}).
In the next section, we apply these general restrictions to 
polynomial splines.

\section{Lower bounds for degree bi-3}
\label{sec:lower}
We now argue that, in general,
\vloc\ enforcement of unbiased $G^1$ constraints
with polynomial tensor-product splines of degree bi-3 
(bicubic) is possible only if the spline patches 
have at least two internal double knots per edge.

Since we specialize to polynomials $\pat{\ii}$ of degree bi-3,
equality in the $G^1$ constraints implies that 
$\alpha$ is a rational function,
$\alpha =: \frac{\numer}{\denom}$.
In fact, we have a low bound on the degrees
of the numerator $\numer$ and the denominator $\denom$.

\begin{lemma} [$\alpha$ degree restricted]
If the two bi-3 patches $\pat{\ii}$ and $\pat{\ii-1}$ 
satisfy an unbiased $G^1$ constraint \eqref{eq:g1} then either
\begin{align}
   &\alpha^k := \frac{\numer}{\denom} \text{ is rational with } \\
   \notag
   &(\dg{\numer},\dg{\denom}) \in \{(2,1),(2,0),(1,1),(1,0),(0,1),(0,0)\}
   \\
   &\text{ and }\quad
   \partial_1\pat{\ii}(u,0) = \lin(u)\denom(u),
   \dg{\lin} \le 2-\dg{\denom}
\end{align}
or the boundary $\pat{\ii}(u,0)$ is forced to have a straight segment.
\label{lem:alpha_rational}
\end{lemma}

\begin{proof}
We may assume that $\numer$ and $\denom$ are relatively coprime.
Since the left hand side
$\partial_2\pat{\ii}(u,0)+\partial_1\pat{\ii-1}(0,u)$
of the $G^1$ constraint \eqref{eq:g1} is polynomial,
$\denom(u)$ must be a (scalar) factor of 
$\partial_1 \pat{\ii}(u,0) \in \R^3$,
the (vector-valued) derivative of the boundary curve.
Unless $\pat{\ii}(u,0)$ is a line segment,
$0<\dg{\partial_1 \pat{\ii}(u,0)}\le 2$.
Consequently $\dg{\denom} \le 2$ and since $\dg{\denom} = 2$ implies that
$\partial_1 \pat{\ii}(u,0) = \bv\denom$ for a constant $\bv\in\R^3$,
$\dg{\denom}\le 1$ must hold to avoid that $\pat{\ii}(u,0)$ is a straight
segment.
Since 
$\dg{\partial_2\pat{\ii}(u,0)+\partial_1\pat{\ii-1}(0,u)}\le 3$,
also
$\dg{\partial_1 \pat{\ii}(u,0) \numer}\le 3$ and therefore 
$\dg{\numer}\le 2$.
\end{proof}

After scaling numerator and denominator,
we may assume that $\denom(u) := 1+\denom_1 u$.
Not linear $\alpha$ then forces a particular 
boundary curve.
\begin{corollary} [$\alpha$ not linear restricts boundary curves]
If $(\dg{\numer},\dg{\denom}) \in \{(2,1),(2,0),(1,1),(0,1)\}$
then the corresponding degree 3 boundary curve segment 
is of the form \eqref{eq:lin_den}.
\label{cor:notlinear}
\end{corollary}
\begin{proof}
The derivative of the curve segment either has a linear 
factor $\denom$ or it is linear because $\dg{\numer}=2$.
\end{proof}

Lemma \ref{lem:quad} and Corollary \ref{cor:notlinear} together
imply that in general, at end points, $\alpha$ must be linear or constant
if we require more flexibility than forced straight line segments.

\begin{corollary} [$\alpha$ not linear at higher-order saddle]
If $\pat{\ii}(u,0)$ emanates from a symmetric higher-order saddle point
then $\alpha^\ii$ in the unbiased $G^1$ constraints \eqref{eq:g1}
must be linear or constant for $\pat{\ii}(u,0)$ not to be a straight segment.
\label{cor:vert_seg}
\end{corollary}

\begin{figure}[h]
    \centering
    \psfrag{b1}[l]{$\pat{1}$}
    \psfrag{b2}{$\pat{2}$}
    \psfrag{b3}{$\pat{3}$}
    \psfrag{b4}[l]{$\pat{4}$}
    \psfrag{a1}{$\alpha^1$}
    \psfrag{a2}{$\alpha^2$}
    \psfrag{a3}{$\alpha^3$}
    \psfrag{s1}{$_1$}
    \psfrag{s2}{$_2$}
    \epsfig{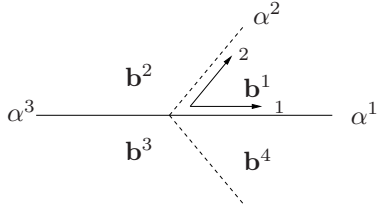}
    \caption{ (Figure \ref{fig:spl2} repeated)
    Join across an {\bf \ik } on the boundary 
    (solid) between two splines.
    The first spline has polynomial pieces $\pat{1}$ and $\pat{2}$.
    }
    \label{fig:spl22}
\end{figure}

The next lemma shows that at \ik s,
neighboring pieces of $\alpha$ constrain one another
more than just by \eqref{eq:dd} and \eqref{eq:ddd}.

\begin{lemma} [$\alpha$ not linear at single knot]
Let the segments be arranged as in Figure \ref{fig:spl22}
(the same as Figure \ref{fig:spl2}),
the \ik\ \emph{single}
and the left boundary segment 
($\pat{3}(u,0)$ shared by the two bi-3 splines)
fixed but general (in the sense that the control points
cannot be assumed to be in a particular relation).
Then $\alpha^1$ can only be not linear if 
\begin{equation}
   \al{3}{0}= 0,  \al{3}{1}=0, 
   \text{ and } 
   \al{3}{2} = \al{1}{2} 
   \ne 0
   .
\end{equation}
In particular, $\alpha^3$ must also be not linear. 
\label{lem:notlinear}
\end{lemma}
\begin{proof}
If $\alpha := \alpha^1$ is not linear 
then Lemma \ref{lem:alpha_rational} implies
$
   (\dg{\numer},\dg{\denom}) \in \{(2,1),(2,0),(1,1),(0,1)\}
$
and therefore $\partial_1\pat{1}(u,0) :=  \lin(u)\denom(u)$,
a linear vector-valued polynomial times the scalar
(possibly constant) factor $\denom(u) := 1+\denom_1 u$.
By \eqref{eq:al10} and \eqref{eq:al11} and the $C^2$ constraints
for the boundary curve, constraint \eqref{eq:ddd} becomes
\begin{equation}
   \text{ at } (0,0),\quad
   0 =
   \al{3}{0} ( \underbrace{\partial^3_1 \pat{3} - 2\denom_1\lin'(0)}_{=:\bv})
   + 4 \al{3}{1} \partial^2_1 \pat{3} 
   + (\al{3}{2} - \al{1}{2})\tn^3
   .
   \label{eq:b3}
\end{equation}
By $C^1$ continuity
$\lin(0)\denom(0) = \lin(0) = -\tn^3$ and hence the $C^2$ constraint
$\partial^2_1 \pat{3} = 
\lin(0)\denom_1 + \lin'(0) = -\tn^3\denom_1 + \lin'(0)$
implies
\begin{equation}
   \lin'(0) = \tn^3\denom_1+\partial^2_1 \pat{3}(0,0)
   .
\end{equation}
Therefore, at $(0,0)$,
$\bv = \partial^3_1 \pat{3} - 2\denom_1(\tn^3\denom_1+\partial^2_1
\pat{3})$.
Since, in general,
$\partial^3_1 \pat{3}(0,0)$,
$\partial^2_1 \pat{3}(0,0)$ and $\tn^3$ are linearly independent,
the scalar $\denom_1$ can not
force $\bv=0$ (recall that $\pat{3}$ is fixed),
and since 
$\bv$, $\partial^2_1 \pat{3}(0,0)$ and $\tn^3$ are linearly independent,
we must have $\al{3}{0}=0$ and $\al{3}{1}=0$ 
and $\al{1}{2}= \al{3}{2}$ in order for \eqref{eq:b3} to hold.

If $\alpha^3$ is linear then $\al{3}{2}=0$ and since
$\alpha''(0) = \left(\frac{\numer}{\denom}\right)''(0) = \numer''(0)$
when $\alpha(0)=\alpha'(0)=0$
(note that $\denom(0)=1$ and hence $\numer(0)=\numer'(0)=0$),
we have $\alpha^1\equiv0$ contradicting
the assumption that $\alpha^1$ is not linear.
\end{proof}

We now have all the pieces in place to prove the main theorem
of smooth surface construction with bi-3 splines.

\begin{theorem} [two double \ik s needed]
In general, using splines of degree bi-3 for
a \vloc\ unbiased $G^1$ construction 
without forced linear boundary segments 
requires the splines to have at least two internal double knots.  
\label{thm:twodblknt}
\end{theorem} 
\begin{proof}
In general, if the boundary curve has only 
a single 1-fold knot (hence two $C^2$-connected segments)
there are not enough degrees of freedom to enforce $C^2$ continuity
of the piecewise curve.
If there are two 1-fold knots (three $C^2$-connected segments),
$C^2$ continuity uniquely determines all boundary coefficients.
If there is one 2-fold knot (two $C^1$-connected segments),
$C^1$ continuity uniquely determines all boundary coefficients.
However, in these last two cases, \eqref{eq:ddd} is unresolved 
at the (two, respectively one) \ik s $\{\tau_i\}$ and
therefore, in general, these base cases allow for constructing a 
$C^2$ boundary curve but not for enforcing \eqref{eq:g1}.

Inserting one additional \ik\ that is 1-fold 
creates one additional boundary curve segment $j$
of degree 3 constrained by four vector-valued
constraints: the parametric $C^0$, $C^1$ and $C^2$ constraints
plus \eqref{eq:ddd}
or, equivalently, one free spline control point subject to \eqref{eq:ddd}.
If $\alpha_j$ is linear, its two coefficients are determined
via \eqref{eq:al10} and \eqref{eq:al11} by those of the neighbor segment,
and therefore the free (B-spline) control point
must be used to resolve \eqref{eq:ddd}.
That is, if $\alpha_j$ is linear,
we do not gain degrees of freedom that would enable
enforcing \eqref{eq:ddd}
at the \ik s $\{\tau_i\}$ of the base case.


By Corollary \ref{cor:vert_seg}, the starting segment's
$\alpha_0$ can be assumed to be linear.
Let $\alpha_j$ be not linear while $\alpha_l$, $l=0,\ldots,j-1$,
$j\ge 1$, are linear. 
By the reasoning of the previous paragraph
all $\pat{l}(u,0)$, $l=0,\ldots,j-1$ are determined
so that Lemma \ref{lem:notlinear} applies:
that is, $\alpha_j$ can only be not linear 
if there is at least by one double knot
between some segment $\pat{l-1}(u,0)$ and $\pat{l}(u,0)$. 

The symmetric argument at the other end implies the claim.
\end{proof}

The proof of Theorem \ref{thm:twodblknt} reveals slightly more than 
its claim: the interior segment with $\alpha_j$ not linear
must be separated by double knots from either end segment.
The simplest such construction is then based on
three segments with the middle segment bracketed by two double knots,
and such that 
$\alpha_0$ and $\alpha_2$ are linear and $\alpha_1$ quadratic
(see \Fig{fig:spl2}, \IR).

\begin{corollary}[lower bound is sharp]
The construction in \cite{Fan:2008:SBS} uses the fewest
knots when creating a smooth surface without forced linear segments
with one bi-3 spline associated with each quad of a general quad mesh.
\label{cor:opt}
\end{corollary}
\begin{proof}
By covering each quad with a $3\times3$ arrangement of 
parametrically $C^1$-connected bi-3 patches in Bernstein-B\'ezier-form,
the construction in \cite{Fan:2008:SBS} uses exactly two \ik s,
both 2-fold. By its choice of quadratic $\alpha_1$
just for the $G^1$ constraints across the middle segment
and linear $\alpha_0$ and $\alpha_2$ for the end segments,
it does not have the shape problem
characterized by Lemma \ref{lem:quad}.
\end{proof}

\begin{figure}[h]
    \centering
    \epsfig{file=\CAGDfig/tricube.eps,width=0.4\wid}
    \hskip 0.1\wid
    \raisebox{0.5cm}{
    \epsfig{file=\CAGDfig/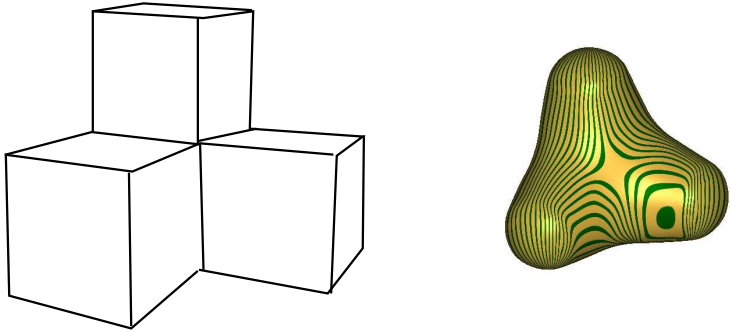,width=0.3\wid}}
    \caption{
    {\bf No Shape defect} (no forced straight line segments) 
    in a higher order saddle (cf.\ Figure \ref{fig:defect}).
    }
    \label{fig:nodefect}
\end{figure}

\noindent



\section{Discussion and Conclusion}
\label{sec:discuss}
%
Remarkably, the results in 
Section \ref{sec:g1} do not depend on the degree or 
even the polynomial nature of splines, but 
assume only sufficiently smooth functions that are 
piecewise with smooth transitions between the pieces.
In particular, the results apply to finite refinement 
by subdivision which creates parametrically smooth transitions
within each \genspl.
The extension to \genspl s mapping
to $\R^d$, $d>3$ is straightforward.

For bi-3 splines these general constraints
imply a lower bound on the number and distribution of knots.
The construction in \cite{Fan:2008:SBS} shows the
lower bound to be tight.

The results extend to constructions based on $G^1$ transitions
of the form
$
\beta^\ii(u) \partial_2\pat{\ii}(u,0)
+
\gamma^\ii(u) \partial_1\pat{\ii-1}(0,u)
=
\alpha^{\ii}(u) \partial_1 \pat{\ii}(u,0) 
$
for which there is a sufficiently rich set of 
input data that imply $\beta=\gamma$.
For example, if $(\alpha^{\ii},\beta^\ii,\gamma^\ii)$
reflect the local geometric distribution of the input data, 
any locally symmetric input yields $\beta=\gamma$
and the results of the paper hold.
%


\medskip
The bounds provide a checklist for constructions.
Theorem \ref{thm:twodblknt} 
implies for example that there is a subtle error 
in the proof of the non-trivial construction \cite{conf/gmp/HahmannBC08}
which uses one double \ik\ only: the construction
falls foul of Corollary \ref{cor:valfour}. 
Such a $2\times2$ split construction can only succeed in special cases.
Choosing generic input data and $\val^1=\val^2=\val^3=4$ 
but $\val^4=3$ shows the problem.
As a second example, Lemma \ref{lem:vloc} prevents a \vloc\ solution
with all $\alpha_{j}$ linear.
When this lemma is specialized by fixing the degree to be 3,
by increasing the patch continuity to $C^2$ and by choosing 
$\alpha_{0j} := \frac{q-j}{q}\alpha_{00} + \frac{j}{q}\alpha_{0q}$
then it yields a proof of the claim 
\cite[Thm 3.1]{journals/cagd/ShiWWL04}.
(In light of \eqref{eq:ddd}, we might adjust the 
titles of \cite{journals/cagd/ShiWWL04} and\cite{shi2006rcg} since 
we cannot have $G^1$ surfaces when adding single knots.)

\medskip
When we restrict connectivity,
i.e.\ drop the assumption made at the outset
that the construction applies to general input 
and uses one tensor-product spline per quad,
then constructions with fewer \ik s are
possible.
For restricted connectivity, it is well known that if all valences
are odd or tangents are in 
an X configuration, then vertex-enclosure does not impose constraints
and simple B\'ezier constructions are possible
(e.g.\ \cite{Wijk:1986:BPA,Peters:1991:SIM,Gregory:1994:FPH}).
If $n^0=n^1$ always holds, say when smoothing a cube,
then we can choose linear $\alpha^1$
and $\alpha^3$ with $\al{1}{1}=\al{3}{1}$ and $\al{1}{0}=0$
to enforce \eqref{eq:dd}.
That is, a construction with one double \ik\ is possible.
Such a construction, covering a quad by $2\times2$ bi-3 patches,
is proposed in \cite{conf/gmp/HahmannBC08}.
A similar but dual, spline-like construction appears in \cite{Zhao:1995:RBS}.
%
Global constructions, singular parameterization,
or control of the valence, for example by splitting patches, 
can allow for structurally or degree-wise simpler constructions, e.g.\
\cite{journals/cagd/Reif95a},
\cite[9.11]{Prau:BBT:2002}, \cite{Peters:1991:SIM,Peters:1995:SS}.

If we allow higher degree, then general
constructions of smooth surfaces with one patch per quad
are shown possible for degree bi-5, for example 
\cite{Myles:2008:GCQ}. For degree bi-4, a single knot (a 2x2-split)
must be introduced (see e.g.\ \cite{Peters:1995:BSS}).

The case of several $G^1$-connected patches per quad 
still awaits full investigation, as does the case of rational
bi-3 patches and the generalization of the problem to unbiased 
$G^k$ transitions for $k>1$.








\bibliographystyle{alpha}
\bibliography{p}

\end{document}